\documentclass[11pt, a4paper]{article}

\usepackage{amsmath}
\usepackage{amssymb}
\usepackage{amsfonts}
\usepackage{amsthm}
    \newtheorem{theorem}{Theorem}
    
    \newtheorem{lemma}[theorem]{Lemma}
    \newtheorem{corollary}[theorem]{Corollary}
    
    \newtheorem*{fact*}{Fact}
    \newtheorem{observation}{Observation}
    \newtheorem*{observation*}{Observation}
    \newtheorem{claim}{Claim}
    \newtheorem*{claim*}{Claim}
    \theoremstyle{definition}

    \newtheorem*{remark*}{Remark}
\usepackage{graphicx}
\usepackage{fullpage}
\usepackage{hyperref}
\usepackage{cite}
\usepackage{algorithmic}
\usepackage{algorithm}
\usepackage{comment}

\title{Validating a PTAS for Triangle-Free 2-Matching \\
via a Simple Decomposition Theorem%
\thanks{%
This work was partially supported by the joint project of Kyoto University and Toyota Motor Corporation,
titled ``Advanced Mathematical Science for Mobility Society'', and
by JSPS KAKENHI Grant Numbers JP22H05001 and JP24K02901. 
}
}

\author{
Yusuke Kobayashi\thanks{Research Institute for Mathematical Sciences, Kyoto University.
E-mail: \{yusuke, tnoguchi\}@kurims.kyoto-u.ac.jp}
\and Takashi Noguchi\footnotemark[2]
}
\date{}

\begin{document}

\maketitle

\begin{abstract}
A triangle-free (simple) 2-matching is an edge set that has at most $2$ edges incident to each vertex and  contains no cycle of length $3$.
For the problem of finding a maximum cardinality triangle-free 2-matching in a given graph, a complicated exact algorithm was proposed by Hartvigsen.
Recently, a simple PTAS using local search was presented by Bosch-Calvo, Grandoni, and Ameli, but its validity proof is not easy.
In this paper, we show a natural and simple decomposition theorem for triangle-free 2-matchings, which leads to a simpler validity proof of the PTAS for the problem.
\end{abstract}

\section{Introduction}
For an undirected graph $G=(V, E)$, an edge subset $M\subseteq E$ is called a \emph{(simple) 2-matching}\footnote{In this paper, we only deal with simple 2-matchings, and so “simple” is omitted for brevity.} if the number of edges in $M$ incident to each vertex is at most $2$. 
It is known that a 2-matching of maximum cardinality for a given graph can be computed in polynomial time by using a matching algorithm, but the problem becomes difficult when additional constraints are imposed.

We say that a 2-matching $M$ is \emph{triangle-free} if it contains no cycle of length $3$, which is called a \emph{triangle}. 
In this paper, we consider the problem of finding a triangle-free 2-matching of maximum cardinality, which we call {\sc Triangle-Free 2-Matching}.
\medskip

\noindent\underline{{\sc Triangle-Free 2-Matching}}\ \ \ 
Given a simple graph $G = (V, E)$, find a triangle-free 2-matching $M \subseteq E$ of maximum cardinality. 
\medskip

A polynomial-time algorithm for {\sc Triangle-Free 2-Matching} was proposed by Hartvigsen in 1984~\cite{HartD}, and its improved version was recently published~\cite{R-HartD}. 
However, his algorithm and analysis are quite complex. 
Recently, Paluch~\cite{PalK} also reported another polynomial-time algorithm for the problem, which is also complex and difficult to verify. 
{\sc Triangle-Free 2-matching} has an application in approximation algorithms for survivable network design problems as described later, and we can use a PTAS insead of an exact algorithm in this context.
With this motivation, a simple PTAS for {\sc Triangle-Free 2-Matching} was proposed by Bosch-Calvo, Grandoni, and Ameli~\cite{BGA}.
For $\varepsilon >0$, their $(1-\varepsilon)$-approximation algorithm is a local search algorithm described as follows. 
We begin with an initial solution $APX=\emptyset$. 
In each iteration, we find a trail $P$ with $|P|\le 2/\varepsilon$ satisfying that 
\begin{equation}\label{condition:augment}
    APX\bigtriangleup P\ \text{is a triangle-free 2-matching whose size is larger than}\ |APX|, \tag{$\ast$}
\end{equation}
and update $APX$ to $APX\bigtriangleup P$.
Here, $\bigtriangleup$ denotes the symmetric difference operator.
If there does not exist such a trail $P$, we return $APX$ (see also Algorithm~\ref{alg:unweghted}).

\begin{algorithm}
    \caption{$(1-\varepsilon)$-approximation algorithm for {\sc Triangle-Free 2-Matching}}
    \label{alg:unweghted}
    \begin{algorithmic}
    \STATE $APX \leftarrow \emptyset$
    \WHILE{$\exists$ a trail $P$ with $|P|\le 2/\varepsilon$ satisfying the condition \eqref{condition:augment}}
    \STATE $APX \leftarrow APX \bigtriangleup P$
    \ENDWHILE
    \RETURN $APX$
    \end{algorithmic}
\end{algorithm}

The validity of Algorithm~\ref{alg:unweghted} can be proved by showing the existence of a desired trail when the current solution is not a $(1-\varepsilon)$-approximate solution.
To this end, Bosch-Calvo et al.\ showed technical lemmas (\!\!\cite[Lemmas 3 and 4]{BGA}) that roughly state that the symmetric difference of the current solution and an optimal solution contains edge-disjoint trails with several technical conditions.
Their argument is fundamental, but the proofs of these lemmas require a lot of case analysis. 

In this paper, instead of \cite[Lemmas 3 and 4]{BGA}, we prove the following decomposition theorem for triangle-free 2-matchings, which is the main contribution of this paper.

\begin{theorem}[Decomposition for triangle-free 2-matchings]
\label{thm:tri-part}
    Let $G$ be a simple graph, and $A_1$ and $A_2$ be triangle-free 2-matchings in $G$. 
    Then there is a partition $\mathcal{P}$ of $A_1\bigtriangleup A_2$ into alternating trails w.r.t.\ $(A_1, A_2)$ such that $A_i\bigtriangleup P$ is a triangle-free 2-matching for $i=1, 2$ and for any $P\in \mathcal{P}$.
\end{theorem}

\noindent Here, an \emph{alternating trail w.r.t.\ $(A_1, A_2)$} is a trail traversing edges in $A_1\setminus A_2$ and $A_2\setminus A_1$ alternately.
In the same way as the argument in \cite{BGA}, we can prove the validity of Algorithm~\ref{alg:unweghted} by using Theorem~\ref{thm:tri-part} as follows.

\begin{corollary}
\label{cor:PTAS}
     Given a constant $\varepsilon>0$, Algorithm~\ref{alg:unweghted} computes a $(1-\varepsilon)$-approximate solution for {\sc Triangle-Free 2-Matching} in polynomial time.
\end{corollary}

\begin{proof}
    Let $OPT$ be one of the optimal solutions for {\sc Triangle-Free 2-Matching}. 
    It suffices to show that for any triangle-free 2-matching $APX$ with $|APX|<(1-\varepsilon)|OPT|$, there exists a trail $P$ with $|P|\le2/\varepsilon$ satisfying the condition \eqref{condition:augment} in the algorithm. 
    By applying Theorem~\ref{thm:tri-part} with $A_1=OPT$ and $A_2=APX$, we see that there exists a partition $\mathcal{P}$ of $OPT\bigtriangleup APX$ into alternating trails w.r.t.\ $(OPT, APX)$ such that $APX\bigtriangleup P$ is a triangle-free 2-matching for any $P\in \mathcal{P}$. 
    Since $|APX\bigtriangleup P|-|APX|\in\{0, \pm 1\}$ holds for any $P\in \mathcal{P}$, there exist at least\footnote{In fact, the number of such trails is exactly $|OPT|-|APX|$. This is because $|APX\bigtriangleup P|-|APX| = |OPT| - |OPT\bigtriangleup P| \ge 0$ for each $P \in \mathcal{P}$, where the inequality is by the fact that $OPT \bigtriangleup P$ is a triangle-free 2-matching.} $|OPT|-|APX|$ edge-disjoint alternating trails $P$ w.r.t.\ $(APX, OPT)$ in $APX\bigtriangleup OPT$ satisfying $|APX\bigtriangleup P|-|APX|=1$, which satisfy the condition \eqref{condition:augment}. 
    Therefore, when $|APX|<(1-\varepsilon)|OPT|$, by an averaging argument, the length of the shortest trail satisfying the condition \eqref{condition:augment} is at most
    $$\frac{|OPT\bigtriangleup APX|}{|OPT|-|APX|}\le \frac{2|OPT|}{\varepsilon |OPT|}=\frac{2}{\varepsilon}.$$
    This completes the proof.
\end{proof}

Note that the roles of $A_1$ and $A_2$ are symmetric in Theorem~\ref{thm:tri-part}, while the current solution and an optimal solution do not play symmetric roles in \cite[Lemmas 3 and 4]{BGA}.
In this sense, we can say that Theorem~\ref{thm:tri-part} is a natural and simple decomposition theorem.
In our proof of Theorem~\ref{thm:tri-part}, we generalize the triangle-free constraint by introducing a list $\mathcal{T}$ of triangles that must not be included in 2-matchings. 
Then we prove the decomposition theorem inductively by transforming the graph and reducing the size of $\mathcal{T}$ (see Section~\ref{sec:induction} for detail).

\paragraph{Related Study}
A generalization of {\sc Triangle-Free 2-Matching} is the $C_{\le k}$-free 2-matching problem: given a graph, find a maximum cardinality 2-matching not containing a cycle of length at most $k$ (referred to as $C_{\le k}$-free 2-matching). 
The case of $k=3$ is {\sc Triangle-Free 2-Matching}, and the case of $k\ge |V|/2$ includes the Hamilton cycle problem.
Since {\sc Triangle-Free 2-Matching} is a relaxation of the Hamilton cycle problem, 
it has been used to design an approximation algorithm for a special case of the traveling salesman problem~\cite{AMP}.
While it is shown by Papadimitriou (described by Cornuéjols and Pulleyblank~\cite{CP}) that the case of $k\ge 5$ is NP-hard, the complexity for the case of $k=4$ is open.

One can consider the weighted version of the $C_{\le k}$-free 2-matching problem: given a graph and a non-negative edge weight, find a $C_{\le k}$-free 2-matching with maximum total weight. 
The case of $k\ge 5$ is obviously NP-hard as with the unweighted version, and the case of $k=4$ is also shown to be NP-hard even if the input is restricted to bipartite graphs~\cite{Kira}. 
The complexity for the case of $k=3$, which is called the weighted triangle-free 2-matching problem is still open.
Hartvigsen and Li~\cite{HL} showed that the weighted triangle-free 2-matching problem can be solved in polynomial time when the input graph is restricted to subcubic graphs, i.e., at most $3$ edges are incident to each vertex.
Simpler polynomial-time algorithms for the problem on subcubic graphs were proposed by Paluch and Wasylkiewicz~\cite{PW} and by Kobayashi~\cite{K2010}.
Futhermore, Kobayashi~\cite{K2022} proposed a polynomial-time algorithm for the problem when all triangles in the input graph are edge-disjoint.

Recently, it was turned out that {\sc triangle-free 2-matching} can be used as a subroutine of an approximation algorithm for the 2-edge-connected spanning subgraph problem (2-ECSS).
In 2-ECSS, given an undirected graph $G=(V, E)$, we find a 2-edge-connected spanning subgraph $(V, F)$ in $G$ with the minimum cardinality of $F$. 
Here, we say that $(V, F)$ is 2-edge-connected if $(V, F\setminus \{e\})$ is connected for any $e\in F$. 
2-ECSS is one of the most fundamental problems in survivable network design: designing a network that is robust against link or vertex failures. 
For 2-ECSS, there is no PTAS unless P=NP~\cite{CL, CristG}, and there have been many studies on approximation algorithms~\cite{KV, CSS, VV, HVV, SV}. 
An approximation algorithm with the factor of $118/89 + \varepsilon$~$(\simeq 1.326)$ was presented by Garg, Grandoni, and Ameli~\cite{GGA}.
Recently Kobayashi and Noguchi~\cite{KN} showed that if a PTAS for {\sc Triangle-Free 2-Matching} exists, then there exists a $(1.3+\varepsilon)$-approximation algorithm for 2-ECSS. 
They introduced the problem of finding a triangle-free 2-edge-cover of minimum cardinality, which has the same complexity as {\sc Triangle-Free 2-Matching}, and  used its solution as a lower bound for 2-ECSS. 
A similar approach has a potential to be used also for a vertex-connectivity variant of 2-ECSS.

\section{Preliminary} 
Let $G=(V,E)$ be a graph that may have self-loops but has no parallel edges. 
For an edge set $F \subseteq E$ and for a vertex $v \in V$, let $d_{F}(v)$ be the degree of $v$ in $(V, F)$, which is the number of the edges in $F$ incident to $v$, where self-loops are counted twice.
A \emph{trail} $P$ (of length $k$) is defined as a sequence of distinct consecutive edges $u_1 u_2, u_2 u_3, \dots , u_k u_{k+1}$. 
Note that $u_i$ and $u_j$ are not necessarily distinct for any $i, j$. 
We sometimes identify a trail $P$ with the corresponding set of edges (ignoring their order). 
A \emph{triangle} is an edge set that forms a cycle of length $3$.
For 2-matchings $A_1$ and $A_2$ in $G$, we say that a trail $P$ is \emph{alternating w.r.t.\ $(A_1, A_2)$} if $P$ traverses edges in $A_1\setminus A_2$ and $A_2\setminus A_1$ alternately.

Before going into the proof of Theorem~\ref{thm:tri-part}, we now consider the case without the triangle-free constraint.
Although the following lemma is a well-known result, we provide a proof for completeness.

\begin{lemma}
\label{lem:base}
    Let $A_1$ and $A_2$ be 2-matchings in $G$. 
    Then there is a partition $\mathcal{P}$ of $A_1\bigtriangleup A_2$ into alternating trails w.r.t.\ $(A_1, A_2)$ such that $A_i\bigtriangleup P$ is a 2-matching for $i=1, 2$ and for any $P\in \mathcal{P}$. 
\end{lemma}

\begin{proof}
    Let $\mathcal{P}$ be a partition of $A_1\bigtriangleup A_2$ into alternating trails w.r.t.\ $(A_1, A_2)$ which has fewest trails among such partitions. 
    Note that such a partition exists because every single edge in $A_1\bigtriangleup A_2$ is an alternating trail. 
    To prove that $\mathcal{P}$ satisfies the condition, we show that $d_{A_i\bigtriangleup P}(v)\le 2$ holds for $i=1,2$, for any $P\in \mathcal{P}$, and for any vertex $v$.
    If $d_{A_1\cap P}(v)=d_{A_2\cap P}(v)$, then $d_{A_i\bigtriangleup P}(v)=d_{A_i}(v)\le 2$.
    
    Now consider the case when $d_{A_1\cap P}(v)>d_{A_2\cap P}(v)$, which implies that $v$ is one of the endpoints of $P$ and the first (or last) edge of $P$ incident to $v$ is in $A_1\setminus A_2$. 
    In this case, $d_{A_1\cap P'}(v)\ge d_{A_2\cap P'}(v)$ holds for any $P'\in \mathcal{P}\setminus \{P\}$. 
    This is because if $d_{A_1\cap P'}(v)>d_{A_2\cap P'}(v)$ holds for some $P'\in \mathcal{P}\setminus \{P\}$, then $v$ is one of the endpoints of $P'$ and the first (or last) edge of $P'$ incident to $v$ is in $A_2\setminus A_1$, which shows that concatenating $P$ and $P'$ leads to a partition with fewer trails than $\mathcal{P}$, contradicting the choice of $\mathcal{P}$. 
    This shows that
    \begin{eqnarray}
        d_{A_2\bigtriangleup P}(v) &=& d_{A_2\setminus P}(v) + d_{A_1\cap P}(v) \nonumber\\
        &=& d_{A_1\cap A_2}(v) + \sum_{P'\in \mathcal{P}\setminus \{P\}}d_{A_2\cap P'}(v) + d_{A_1\cap P}(v) \nonumber\\
        &\le& d_{A_1\cap A_2}(v) + \sum_{P'\in \mathcal{P}\setminus \{P\}}d_{A_1\cap P'}(v) + d_{A_1\cap P}(v) =d_{A_1}(v)\le 2. \nonumber
    \end{eqnarray}
    Since $d_{A_1\cap P}(v)>d_{A_2\cap P}(v)$ holds, we also see that
    
    $$d_{A_1\bigtriangleup P}(v)=d_{A_1\setminus P}(v)+d_{A_2\cap P}(v)<d_{A_1\setminus P}(v)+d_{A_1\cap P}(v)=d_{A_1}(v)\le 2.$$
    
    The same argument works for the case when $d_{A_1\cap P}(v)<d_{A_2\cap P}(v)$ by changing the roles of $A_1$ and $A_2$.

    Therefore, $A_i\bigtriangleup P$ is a 2-matching for $i=1, 2$ and for any $P\in \mathcal{P}$.
\end{proof}

Let $\mathcal{T}$ be a subset of triangles in $G=(V, E)$. 
An edge set $F\subseteq E$ is called $\mathcal{T}$-\emph{free} if $F$ does not contain any triangle $T\in\mathcal{T}$ as a subset.
Note that when $\mathcal{T}$ is the set of all triangles in $G$, $\mathcal{T}$-freeness is equivalent to triangle-freeness.

For an edge set $F\subseteq E$, we define $\mathcal{T}(F)\subseteq \mathcal{T}$ as the set of all triangles in $\mathcal{T}$ containing at least one edge in $F$. 
Since there are no parallel edges, we observe the following, which will be used in our argument. 

\begin{observation}
\label{obs:T(T)}
    Let $\mathcal{T}$ be a subset of triangles, $M$ be a 2-matching, and $T$ be a triangle (possibly, $T\notin \mathcal{T}$).
    If $|M\cap T|=2$ holds then $M$ is $\mathcal{T}(T)$-free.
\end{observation}

\section{Proof of Theorem~\ref{thm:tri-part}}
\label{sec:induction}
In this section, we prove the following theorem, which is stronger than Theorem~\ref{thm:tri-part}.

\begin{theorem}[Decomposition for $\mathcal{T}$-free 2-matchings]
\label{thm:T-part}
    Let $G$ be a graph that may have self-loops but has no parallel edges, $\mathcal{T}$ be a subset of triangles in $G$, and $A_1$ and $A_2$ be $\mathcal{T}$-free 2-matchings in $G$. 
    Then there exists a partition $\mathcal{P}$ of $A_1\bigtriangleup A_2$ into alternating trails w.r.t.\ $(A_1, A_2)$ such that $A_i\bigtriangleup P$ is a $\mathcal{T}$-free 2-matching for $i=1, 2$ and for any $P\in \mathcal{P}$.
\end{theorem}

\begin{proof}
We prove Theorem~\ref{thm:T-part} by induction on $|\mathcal{T}|$.
If $|\mathcal{T}|=0$, then Lemma~\ref{lem:base} immediately shows the existence of $\mathcal{P}$ satisfying the conditions in Theorem~\ref{thm:T-part}. 
This is the base case of the induction.
Suppose that $|\mathcal{T}|>0$ holds. 
The induction step is divided into the following cases.

\paragraph{Case 1.}
    If there exists a triangle $T\in \mathcal{T}$ not contained in $A_1\cup A_2$, then we define a subset $\mathcal{T'}$ of triangles in $G$ as $\mathcal{T'}:=\mathcal{T}\setminus \{T\}$. 
    Since $A_1$ and $A_2$ are both $\mathcal{T'}$-free 2-matchings and $\mathcal{T'}$ satisfies $|\mathcal{T'}|<|\mathcal{T}|$, by the induction hypothesis, one can construct a partition $\mathcal{P}$ of $A_1\bigtriangleup A_2$ into alternating trails w.r.t.\ $(A_1, A_2)$ such that $A_i\bigtriangleup P$ is a $\mathcal{T'}$-free 2-matching for $i=1, 2$ and for any $P\in \mathcal{P}$. 
    Since $A_i\bigtriangleup P$ does not contain $T$, $\mathcal{P}$ satisfies the conditions in Theorem~\ref{thm:T-part}. 

\paragraph{Case 2.}
    Suppose that there exists a triangle $T=\{u_1 u_2, u_2 u_3, u_3 u_1\} \in \mathcal{T}$ with $u_1 u_2\in A_1\setminus A_2, u_2 u_3\in A_1 \cap A_2, u_3 u_1\in A_2 \setminus A_1$. 
    In this case, we perform the following operation.
    First, we construct a graph $G'$ from $G$ by contracting $\{u_2, u_3\}$ into a single vertex $u_{23}$. 
    Parallel edges created by the contraction are identified and treated as a single edge. 
    We consider the following two cases separately. 

\begin{enumerate}
    \item Suppose that there exists a vertex $u'_1$ in $G$ with $u'_1 u_2\in A_2 \setminus A_1$ and $u_3 u'_1\in A_1 \setminus A_2$ (see Figure~\ref{fig:2-1} left). 
    Let $T'$ be a triangle $\{u'_1 u_2, u_2 u_3, u_3 u'_1\}$ (possibly, $T'\notin \mathcal{T}$). 
    Define a subset $\mathcal{T'}$ of triangles in $G'$ and 2-matchings $A_1'$ and $A_2'$ in $G'$ as $\mathcal{T'} :=\mathcal{T}\setminus\mathcal{T}(T\cup T')$, $A_1':=A_1\setminus\{u_2 u_3\}$, and $A_2':=A_2\setminus\{u_2 u_3\}$. 
    Since $A'_1$ and $A'_2$ are both $\mathcal{T'}$-free 2-matchings and $\mathcal{T'}$ satisfies $|\mathcal{T'}|<|\mathcal{T}|$, by the induction hypothesis, one can construct a partition $\mathcal{P'}$ of $A'_1\bigtriangleup A'_2$ into alternating trails w.r.t.\ $(A'_1, A'_2)$ such that $A'_i\bigtriangleup P$ is a $\mathcal{T'}$-free 2-matching for $i=1, 2$ and for any $P\in \mathcal{P'}$.
    Then, $P_1:=u_1 u_2,  u_2 u'_1, u'_1 u_3, u_3 u_1$ is an alternating trail w.r.t.\ $(A_1, A_2)$, and $\mathcal{P}:=\mathcal{P'}\cup\{P_1\}$ is a partition of $A_1\bigtriangleup A_2$ into alternating trails w.r.t.\ $(A_1, A_2)$ (see Figure~\ref{fig:2-1} right).

    To see that $A_i\bigtriangleup P$ is a 2-matching for $P\in \mathcal{P}$, we only need to consider the degree of the endpoints of $P$. 
    Since $P_1$ is a closed alternating trail, $A_i\bigtriangleup P_1$ is clearly a 2-matching. In addition, each $P\in \mathcal{P}\setminus \{P_1\}$ satisfies $d_{A_i\bigtriangleup P}(v)=d_{A'_i\bigtriangleup P}(v)\le 2$ for an endpoint $v$ of $P$, where we note that $v\notin \{u_2, u_3\}$.
    Therefore $A_i\bigtriangleup P$ is a 2-matching for any $P\in \mathcal{P}$.
    Let us show that $A_i\bigtriangleup P$ is $\mathcal{T}$-free. 
    For any $P\in \mathcal{P}\setminus \{P_1\}$, since $A_i\bigtriangleup P \subseteq (A'_i\bigtriangleup P)\cup T\cup T'$ and $A'_i\bigtriangleup P$ is $\mathcal{T'}$-free with $\mathcal{T'} :=\mathcal{T}\setminus\mathcal{T}(T\cup T')$, $A_i\bigtriangleup P$ is $\mathcal{T'}$-free. 
    Similarly, since $A_i\bigtriangleup P_1 \subseteq A'_i\cup T\cup T'$ holds, $A_i\bigtriangleup P_1$ is $\mathcal{T'}$-free. 
    In addition, for any $P\in \mathcal{P}$, since $T\cap(A_i\bigtriangleup P)$ is equal to $\{u_1 u_2, u_2 u_3\}$ or $\{u_2 u_3, u_3 u_1\}$, $A_i\bigtriangleup P$ is $\mathcal{T}(T)$-free by Observation~\ref{obs:T(T)}. 
    Similarly, $A_i\bigtriangleup P$ is $\mathcal{T}(T')$-free as $|T'\cap(A_i\bigtriangleup P)|=2$. 
    Therefore $\mathcal{P}$ satisfies the conditions in Theorem~\ref{thm:T-part}.

    \begin{figure}
    \centering
    \includegraphics[width=155mm]{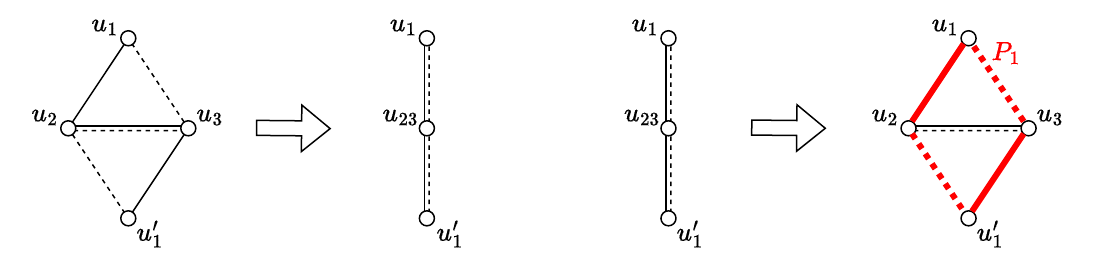}
    \caption{Reduction for Case 2-(1). Solid and dashed lines represent edges belonging to $A_1$ and $A_2$, respectively.}
    \label{fig:2-1}
    \end{figure}
    
    \begin{figure}
    \centering
    \includegraphics[width=155mm]{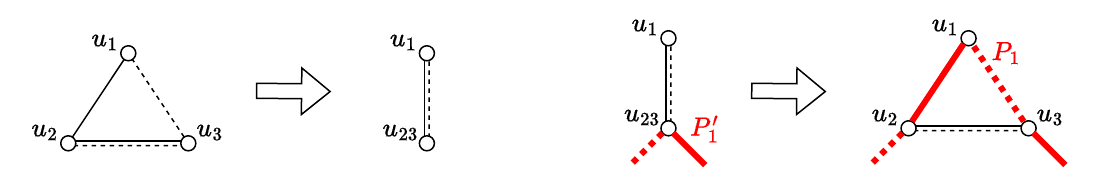}
    \caption{Reduction for Case 2-(2)}
    \label{fig:2-2}
    \end{figure}

    \item Suppose that there does not exist a vertex $u'_1$ in $G$ with $u'_1 u_2\in A_2 \setminus A_1$ and $u_3 u'_1\in A_1 \setminus A_2$ (see Figure~\ref{fig:2-2} left). 
    Define a subset $\mathcal{T'}$ of triangles in $G'$ and 2-matchings $A_1'$ and $A_2'$ in $G'$ as $\mathcal{T'} :=\mathcal{T}\setminus\mathcal{T}(T)$, $A_1':=A_1\setminus\{u_2 u_3\}$, and $A_2':=A_2\setminus\{u_2 u_3\}$. 
    Since $A'_1$ and $A'_2$ are both $\mathcal{T'}$-free 2-matchings and $\mathcal{T'}$ satisfies $|\mathcal{T'}|<|\mathcal{T}|$, by the induction hypothesis, one can construct a partition $\mathcal{P'}$ of $A'_1\bigtriangleup A'_2$ into alternating trails w.r.t.\ $(A'_1, A'_2)$ such that $A'_i\bigtriangleup P$ is a $\mathcal{T'}$-free 2-matching for $i=1, 2$ and for any $P\in \mathcal{P'}$.
    Note that 
    if two distinct trails $Q_1, Q_2 \in \mathcal{P'}$ have $u_{23}$ as an endpoint, then 
    one of $d_{A'_1\bigtriangleup Q_1}(u_{23})$ and $d_{A'_1\bigtriangleup Q_2}(u_{23})$ exceeds $2$, which contradicts the choice of $\mathcal{P'}$.     
    %
    Therefore, in $\mathcal{P'}$, there is at most one trail through vertex $u_{23}$, say $P'_1$. 
    If there is no such a trail, let $P'_1=\emptyset$. 
    Let $P_1$ be the trail in $G$ obtained by inserting $u_2 u_1, u_1 u_3$ in the sequence of edges in $G$ corresponding to $P'_1$. 
    Then, $P_1$ is an alternating trail w.r.t.\ $(A_1, A_2)$, and $\mathcal{P}:=(\mathcal{P'}\setminus\{P'_1\})\cup\{P_1\}$ is a partition of $A_1\bigtriangleup A_2$ into alternating trails w.r.t.\ $(A_1, A_2)$ (see Figure~\ref{fig:2-2} right).

    For any $v\in V\setminus\{u_2, u_3\}$, $d_{A_i\bigtriangleup P}(v)=d_{A'_i\bigtriangleup P}(v)\le 2$ holds for $P\in\mathcal{P}\setminus \{P_1\}$ and $d_{A_i\bigtriangleup P_1}(v)=d_{A'_i\bigtriangleup P'_1}(v)\le 2$ holds for $P_1$.
    We also see that if $v\in V$ is not an endpoint of $P\in \mathcal{P}$, then $d_{A_i\bigtriangleup P}(v)=d_{A_i}(v)\le 2$ holds.
    Suppose that $v\in \{u_2, u_3\}$ is an endpoint of $P\in \mathcal{P}$, which implies that $P=P_1$. 
    By symmetry, we only consider the case when $v=u_2$. 
    Since $P'_1$ has $u_{23}$ as an endpoint, the edges in $A_1\cup A_2$ incident to $u_2$ are only $u_1u_2$ and $u_2u_3$. 
    Thus $d_{A_i\bigtriangleup P_1}(u_2)\le 2$. 
    Therefore, $A_i\bigtriangleup P$ is a 2-matching for any $P\in \mathcal{P}$. 
    Let us show that $A_i\bigtriangleup P$ is $\mathcal{T}$-free. 
    For any $P\in \mathcal{P}\setminus \{P_1\}$, since $A_i\bigtriangleup P \subseteq (A'_i\bigtriangleup P)\cup T$ and $A'_i\bigtriangleup P$ is $\mathcal{T'}$-free, it holds that $A_i\bigtriangleup P$ is $\mathcal{T'}$-free. 
    Similarly, since $A_i\bigtriangleup P_1 \subseteq (A'_i\bigtriangleup P'_1)\cup T$ holds, $A_i\bigtriangleup P_1$ is $\mathcal{T'}$-free. 
    In addition, for any $P\in \mathcal{P}$, $T\cap(A_i\bigtriangleup P)$ is equal to $\{u_1 u_2, u_2 u_3\}$ or $\{u_2 u_3, u_3 u_1\}$, and hence $A_i\bigtriangleup P$ is $\mathcal{T}(T)$-free by Observation~\ref{obs:T(T)}. 
    Therefore, $\mathcal{P}$ satisfies the conditions in Theorem~\ref{thm:T-part}.

\end{enumerate}

\paragraph{Case 3.}
    Suppose that there exists a triangle $T=\{u_1 u_2, u_2 u_3, u_3 u_1\} \in \mathcal{T}$ with $\{u_1 u_2, u_1 u_3\} \subseteq A_1 \setminus A_2, u_2 u_3 \in A_2 \setminus A_1$ such that there does not exist $u'_1\in V\setminus \{u_1\}$ with $\{u'_1 u_2, u_2 u_3, u_3 u'_1\} \in \mathcal{T}$ and $\{u'_1 u_2, u_3 u'_1\} \subseteq A_1$. 
    In this case, we perform the following operation.

    Define a subset $\mathcal{T'}$ of triangles in $G$ and 2-matchings $A_1'$ and $A_2'$ in $G$ as $\mathcal{T'}:=\mathcal{T}\setminus\mathcal{T}(\{u_1 u_2, u_3 u_1\})$, $A_1':=(A_1\setminus\{u_1 u_2, u_3 u_1\})\cup\{u_2 u_3, u_1 u_1\}$, and $A'_2:=A_2$, where $u_1 u_1$ denotes a self-loop incident to $u_1$ (see Figure~\ref{fig:3-1} left and Figure~\ref{fig:3-2} left). 
    If $G$ has no self-loop incident to $u_1$, then a new self-loop is added at $u_1$ in $G$. 
    Since $A_1$ is $\mathcal{T}$-free, $A_1\cup\{u_2 u_3\}$ contains no triangle in $\mathcal{T}\setminus\mathcal{T}(\{u_2 u_3\})$, and hence $A'_1$ is $\mathcal{T}\setminus\mathcal{T}(\{u_2 u_3\})$-free. 
    This shows that $A_1'$ is $\mathcal{T'}$-free because $A'_1$ has no triangle in $\mathcal{T}(\{u_2, u_3\})$ by the assumptions of Case 3. 
    Since $A'_1$ and $A'_2$ are both $\mathcal{T'}$-free 2-matchings, by the induction hypothesis, one can construct a partition $\mathcal{P'}$ of $A'_1\bigtriangleup A'_2$ into alternating trails w.r.t.\ $(A'_1, A'_2)$ such that $A'_i\bigtriangleup P$ is a $\mathcal{T'}$-free 2-matching for $i=1, 2$ and for any $P\in \mathcal{P'}$.

\begin{enumerate}
    \item If $u_1 u_1$ is contained in $A'_2$, then $P_1:=u_1 u_1, u_1 u_2, u_2 u_3, u_3 u_1$ is an alternating trail w.r.t.\ $(A_1, A_2)$, and $\mathcal{P}:=\mathcal{P'}\cup\{P_1\}$ is a partition of $A_1\bigtriangleup A_2$ into alternating trails w.r.t.\ $(A_1, A_2)$ (see Figure~\ref{fig:3-1} right). 

    \begin{figure}
    \centering
    \includegraphics[width=155mm]{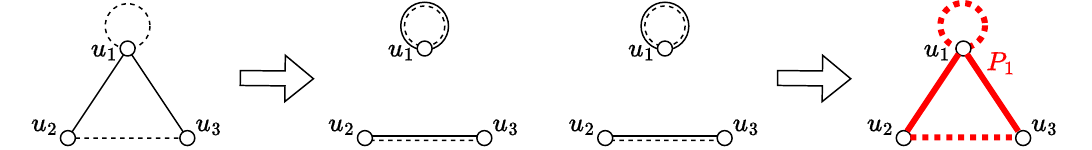}
    \caption{Reduction for Case 3-(1)}
    \label{fig:3-1}
    \end{figure}
    
    Since $P_1$ is a closed alternating trail, clearly $A_i\bigtriangleup P_1$ is a 2-matching. 
    For any $P\in \mathcal{P}\setminus \{P_1\}$ and for any $v\in V$, it holds that $d_{A_i\bigtriangleup P}(v)=d_{A'_i\bigtriangleup P}(v)\le 2$. 
    Thus, $A_i\bigtriangleup P$ is a 2-matching for any $P\in \mathcal{P}$. 
    We will show later together with case (2) that $A_i\bigtriangleup P$ is $\mathcal{T}$-free for $P\in \mathcal{P}$.

    \item If $u_1 u_1$ is not contained in $A'_2$, then there is a trail $P'_1\in \mathcal{P'}$ containing $u_1u_1$. 
    Let $P_1$ be the trail obtained from $P'_1$ by deleting $u_1 u_1$ and inserting $u_1 u_2, u_2 u_3, u_3 u_1$. 
    Then, $P_1$ is an alternating trail w.r.t.\ $(A_1, A_2)$, and $\mathcal{P}:=(\mathcal{P'}\setminus\{P'_1\})\cup\{P_1\}$ is a partition of $A_1\bigtriangleup A_2$ into alternating trails w.r.t.\ $(A_1, A_2)$ (see Figure~\ref{fig:3-2} right).

    \begin{figure}
    \centering
    \includegraphics[width=155mm]{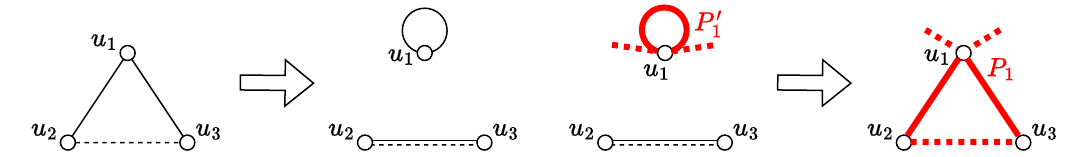}
    \caption{Reduction for Case 3-(2)}
    \label{fig:3-2}
    \end{figure}

    For any $P\in \mathcal{P}\setminus \{P_1\}$ and for any $v\in V$, it holds that $d_{A_i\bigtriangleup P}(v)=d_{A'_i\bigtriangleup P}(v)\le 2$, and hence $A_i\bigtriangleup P$ is a 2-matching. 
    Since $d_{A_i\bigtriangleup P_1}(v)=d_{A'_i\bigtriangleup P'_1}(v)\le 2$ for any $v\in V$, $A_i\bigtriangleup P_1$ is also a 2-matching.

\end{enumerate}
 
    For both of the above cases, let us show that $A_i\bigtriangleup P$ is $\mathcal{T}$-free for $i=1, 2$ and for any $P\in \mathcal{P}$. 
    For any $P\in \mathcal{P}\setminus \{P_1$\}, since $A_i\bigtriangleup P \subseteq (A'_i\bigtriangleup P)\cup\{u_1 u_2, u_3 u_1\}$ and $A'_i\bigtriangleup P$ is $\mathcal{T'}$-free, $A_i\bigtriangleup P$ is $\mathcal{T'}$-free. 
    Similarly, since $A_i\bigtriangleup P_1 \subseteq A'_i\cup \{u_1 u_2, u_3 u_1\}$ holds in case (1) and $A_i\bigtriangleup P_1 \subseteq (A'_i\bigtriangleup P'_1)\cup \{u_1u_2, u_2u_3\}$ holds in case (2), $A_i \bigtriangleup P_1$ is $\mathcal{T'}$-free. 
    In addition, for any $P\in \mathcal{P}$, $T\cap(A_i\bigtriangleup P)$ is equal to $\{u_1 u_2, u_3 u_1\}$ or $\{u_2 u_3\}$. 
    When $T\cap(A_i\bigtriangleup P)=\{u_1 u_2, u_3 u_1\}$, $A_i\bigtriangleup P$ is $\mathcal{T}(\{u_1 u_2, u_3 u_1\})$-free by Observation~\ref{obs:T(T)}. 
    When $T\cap(A_i\bigtriangleup P)=\{u_2 u_3\}$, $A_i\bigtriangleup P$ is $\mathcal{T}(\{u_1 u_2, u_3 u_1\})$-free because it contains neither $u_1 u_2$ nor $u_3 u_1$. 
    Therefore $\mathcal{P}$ satisfies the conditions in Theorem~\ref{thm:T-part}.
    
\paragraph{Case 3'.}
    Suppose that there exists a triangle $T=\{u_1 u_2, u_2 u_3, u_3 u_1\} \in \mathcal{T}$ with $ u_2 u_3 \in A_1 \setminus A_2$ and $\{u_1 u_2, u_3 u_1\} \subseteq A_2 \setminus A_1$ such that there does not exist $u'_1\in V\setminus \{u_1\}$ with $\{u'_1 u_2, u_2 u_3, u_3 u'_1\} \in \mathcal{T}$ and $\{u'_1 u_2, u_3 u'_1\} \subseteq A_2$. 
    Then, partition $\mathcal{P}$ can be constructed by the same argument as in Case 3, where $A_1$ and $A_2$ are swapped. 

\paragraph{Case 4.}
    Suppose that there is no triangle in $\mathcal{T}$ that satisfies any condition in Case 1, 2, 3, or 3'. 
    Since $\mathcal{T}$ is not an empty set,  by changing the roles of $A_1$ and $A_2$ if necessary, there exist two triangles $T=\{u_1 u_2, u_2 u_3, u_3 u_1\} \in \mathcal{T}$ and $T'=\{u'_1 u_2, u_2 u_3, u_3 u'_1\} \in \mathcal{T}$ with $\{u_1 u_2, u_3 u_1, u'_1 u_2, u_3 u'_1\} \subseteq A_1 \setminus A_2$ and $u_2 u_3 \in A_2 \setminus A_1$. 
    Note that $u'_1 u_2, u_3 u'_1 \notin A_2$, because $T'$ does not satisfy the condition in Case 2. 
    In this case, define a subset $\mathcal{T'}$ of triangles in $G$ and 2-matchings $A_1'$ and $A_2'$ in $G$ as $\mathcal{T'}:=\mathcal{T}\setminus\mathcal{T}(T\cup T'), A_1':=(A_1\setminus\{u_3 u_1, u'_1 u_2\})\cup\{u_2 u_3, u_1 u'_1\}$, and $A'_2:=A_2$ (see Figure~\ref{fig:4-1} left and Figure~\ref{fig:4-2} left). 
    If $u_1 u'_1$ does not exist in $G$, then add a new edge $u_1 u'_1$ to $G$. 
    Since $A'_1$ and $A'_2$ are both $\mathcal{T'}$-free 2-matchings and $|\mathcal{T'}| < |\mathcal{T}|$, one can construct a partition $\mathcal{P'}$ of $A'_1\bigtriangleup A'_2$ into alternating trails w.r.t.\ $(A'_1, A'_2)$ such that $A'_i\bigtriangleup P$ is a $\mathcal{T'}$-free 2-matching for $i=1, 2$ and for any $P\in \mathcal{P'}$. 
    For later use, we show the following claim. 

\begin{claim}\label{clm:TT'}
    For any $T''\in \mathcal{T}(T\cup T')$, the vertex set of $T''$ is contained in $\{u_1, u_2, u_3, u'_1\}$.
\end{claim}

\begin{proof}[Proof of the Claim]
    Assume to the contrary that there exists $T''\in \mathcal{T}(T\cup T')$ such that the vertex set of $T''$ contains a vertex $v \notin\{u_1, u_2, u_3, u'_1\}$. 
    Note that $T''$ is contained in $A_1\cup A_2$ because otherwise $T''$ satisfies the condition in Case 1.
    Then $T''$ is denoted by $T''=\{xy, xv, yv\}$, where $xy$ is one of $u_1 u_2, u_3 u_1, u'_1 u_2, u_3 u'_1$, and $u_2 u_3$. 
    If $xy = u_1 u_2$, then $u_1 v$ and $u_2 v$ are in $A_2\setminus A_1$ due to the degree constraint of $A_1$. Therefore, $T''$ satisfies the condition in Case 3', which contradicts the assumption in Case 4. 
    The same argument applies when $xy = u_3 u_1, u'_1 u_2$, or $u_3 u'_1$. If $xy = u_2 u_3$, then $u_2 v$ and $u_3 v$ are in $A_2\setminus A_1$ due to the degree constraint of $A_1$. 
    Then, $T''$ is contained in $A_2$, which contradicts that $A_2$ is $\mathcal{T}$-free. 
\end{proof}

    Then, the case is divided according to whether $A_2$ contains an edge $u_1 u'_1$ or not.

\begin{enumerate}
    \item Suppose that $A_2$ contains $u_1u'_1$. 
    Then, $P_1:=u_1 u_3, u_3 u_2, u_2 u'_1, u'_1 u_1$ is an alternating trail w.r.t.\ $(A_1, A_2)$, and $\mathcal{P}:=\mathcal{P'}\cup\{P_1\}$ is a partition of $A_1\bigtriangleup A_2$ into alternating trails w.r.t.\ $(A_1, A_2)$ (see Figure~\ref{fig:4-1} right).

    \begin{figure}
    \centering
    \includegraphics[width=155mm]{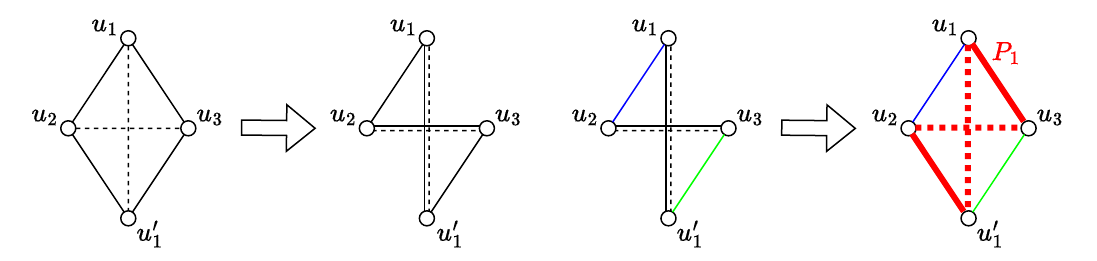}
    \caption{Reduction for Case 4-(1)}
    \label{fig:4-1}
    \end{figure}

    It is easy to see that $A_i\bigtriangleup P$ is a $\mathcal{T'}$-free 2-matching for $i=1, 2$ and for any $P\in \mathcal{P}$.
    In addition, since $\{u_1 u'_1, u'_1 u_2, u_2 u_3, u_3 u_1\}$ is contained in a single trail $P_1$, for $i=1, 2$ and for any $P\in\mathcal{P}$, $A_i\bigtriangleup P$ satisfies $(A_i\bigtriangleup P) \cap \{u_3 u_1, u'_1 u_2\}=\emptyset$ or $(A_i\bigtriangleup P) \cap \{u_2 u_3, u_1 u'_1\} =\emptyset$. 
    Therefore $A_i\bigtriangleup P$ contains none of $T$, $T'$, $\{u_1 u'_1, u'_1 u_2, u_2 u_1\}$, and $\{u_1 u'_1, u'_1 u_3, u_3 u_1\}$, and so $A_i\bigtriangleup P$ is $\mathcal{T}(T\cup T')$-free by Claim~\ref{clm:TT'}.
    Thus $\mathcal{P}$ satisfies the conditions in Theorem~\ref{thm:T-part}.
    
    \item Suppose that $A_2$ does not contain $u_1 u'_1$. 
    Then, in $\mathcal{P'}$, there is a trail $P'_1$ containing $u_1 u'_1$. 
    Let $P_1$ be the trail obtained from $P'_1$ by deleting $u_1 u'_1$ and inserting $u_1 u_3, u_3 u_2, u_2 u'_1$. 
    Then $P_1$ is an alternating trail w.r.t.\ $(A_1, A_2)$, and $\mathcal{P}:=(\mathcal{P'}\setminus\{P'_1\})\cup\{P_1\}$ is a partition of $A_1\bigtriangleup A_2$ into alternating trails w.r.t.\ $(A_1, A_2)$ (see Figure~\ref{fig:4-2} right).

    \begin{figure}
    \centering
    \includegraphics[width=155mm]{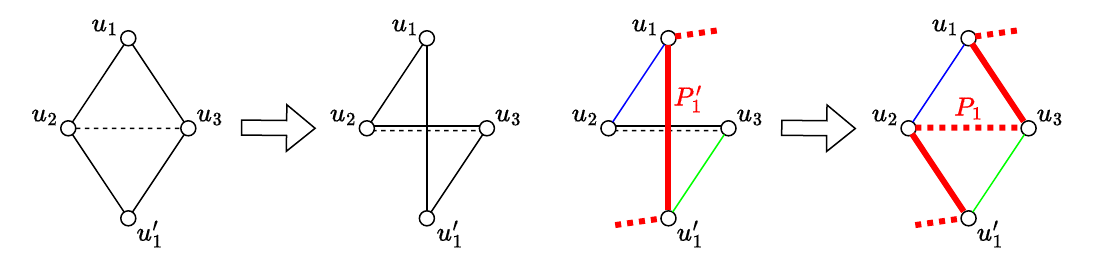}
    \caption{Reduction for Case 4-(2)}
    \label{fig:4-2}
    \end{figure}

    It is easy to see that $A_i\bigtriangleup P$ is a $\mathcal{T'}$-free 2-matching for $i=1, 2$ and for any $P\in \mathcal{P}$.
    In addition, since $\{u_1 u_3, u_3 u_2, u_2 u'_1\}$ is contained in a single trail $P_1$, for $i=1, 2$ and for any $P\in\mathcal{P}$, $A_i\bigtriangleup P$ satisfies $ (A_i\bigtriangleup P) \cap \{u_3 u_1, u'_1 u_2\} =\emptyset$ or $u_2 u_3\notin A_i\bigtriangleup P$. 
    Therefore $A_i\bigtriangleup P$ contains neither $T$ nor $T'$, and so $A_i\bigtriangleup P$ is $\mathcal{T}(T\cup T')$-free by Claim~\ref{clm:TT'}.
    Thus $\mathcal{P}$ satisfies the conditions in Theorem~\ref{thm:T-part}.

\end{enumerate}

Therefore, by the induction, there exists a partition $\mathcal{P}$ of $A_1\bigtriangleup A_2$ into alternating trails w.r.t.\ $(A_1, A_2)$ such that $A_i\bigtriangleup P$ is a $\mathcal{T}$-free 2-matching for $i=1, 2$ and for any $P\in \mathcal{P}$. 
\end{proof}

\vskip\baselineskip
Theorem~\ref{thm:tri-part} follows immediately from Theorem~\ref{thm:T-part} by setting $\mathcal{T}$ as the set of all triangles in $G$.

\section{Concluding Remarks}
In this paper, we have proved a decomposition theorem for triangle-free 2-matchings (Theorem~\ref{thm:tri-part}).
In fact, we have proved a stronger theorem (Theorem~\ref{thm:T-part}), which leads to a similar validity proof for a PTAS for a generalized problem.

\begin{corollary}
    There is a PTAS for the problem of finding a maximum cardinality $\mathcal{T}$-free 2-matching for a given graph $G$ and a triangle set $\mathcal{T}$ in $G$.
\end{corollary}

A possible direction of future research is to extend our result to the edge weighted variant of {\sc Triangle-Free 2-Matching}.
We believe that a more refined decomposition theorem will show the existence of a PTAS also for the weighted version of the problem.

\bibliographystyle{abbrv}
\bibliography{ref}
\end{document}